\documentclass[
  10pt,
  nofootinbib,
  superscriptaddress,
  onecolumn,
  amsfonts,
  amssymb,
  amsmath,
  aps,
  pra,
  floatfix
]{revtex4-2}
\usepackage{graphicx}
\usepackage{mathtools}
\usepackage{mathrsfs}
\usepackage{amsthm}
\usepackage{bm}
\usepackage{booktabs}
\usepackage{array}
\usepackage{longtable}
\usepackage[inline]{enumitem}
\usepackage{csquotes}
\usepackage{comment}
\PassOptionsToPackage{colorlinks=true,citecolor=cyan,urlcolor=cyan,linkcolor=cyan,breaklinks=true}{hyperref}
\usepackage{orcidlink}
\usepackage{hyperref}
\usepackage{url}

\makeatletter
\def\switch@array{}
\makeatother

\newtheoremstyle{break}
  {}{}{\itshape}{}{\bfseries}{.}{\newline}
  {\thmname{#1}\thmnumber{ #2}\thmnote{ (\bfseries #3)}}
\theoremstyle{break}
\newtheorem{theorem}{Theorem}

\newtheorem{definition}{Definition}
\newtheorem{corollary}{Corollary}
\newtheorem{proposition}{Proposition}
\newtheorem{remark}{Remark}

\newcommand{\eq}[1]{\begin{align} #1 \end{align}}

\newcommand{\bbZ}{\mathbb{Z}}
\newcommand{\commutes}{\mathbin{\rotatebox[origin=c]{270}{$\multimap$}}}
\DeclareMathOperator{\com}{com}
\newcommand{\MO}{\mathrm{MO}}
\newcommand{\bvec}[1]{\mathbf{#1}}
\newcommand{\impS}{\Rightarrow_{\mathrm{S}}}
\newcommand{\impC}{\Rightarrow_{\mathrm{C}}}
\newcommand{\impR}{\Rightarrow_{\mathrm{R}}}
\newcommand{\dcS}{\otimes_{\mathrm{S}}}
\newcommand{\dcC}{\otimes_{\mathrm{C}}}
\newcommand{\dcR}{\otimes_{\mathrm{R}}}

\begin{document}

\title{Quantum Logic as the Logic of Contexts}

\author{Haruki Emori\,\orcidlink{0009-0007-2264-9192}}
\email{emori.haruki.i8@elms.hokudai.ac.jp}
\affiliation{Graduate School of Information Science and Technology, Hokkaido University, Kita 14, Nishi 9, Kita-ku, Sapporo, Hokkaido 060-0814, Japan}
\affiliation{RIKEN Center for Interdisciplinary Theoretical and Mathematical Sciences (iTHEMS), 2-1 Hirosawa, Wako, Saitama, 351-0198, Japan}
\affiliation{Teikyo University Advanced Comprehensive Research Organization (ACRO), 2-21-1 Kaga, Itabashi-ku, Tokyo, 173-0003, Japan}

\author{Atsushi Iriki\,\orcidlink{0000-0002-5262-163X}}
\email{iriki.atsushi.lh@teikyo-u.ac.jp}
\affiliation{Teikyo University Advanced Comprehensive Research Organization (ACRO), 2-21-1 Kaga, Itabashi-ku, Tokyo, 173-0003, Japan}
\affiliation{RIKEN Center for Interdisciplinary Theoretical and Mathematical Sciences (iTHEMS), 2-1 Hirosawa, Wako, Saitama, 351-0198, Japan}
\affiliation{Brain Mind and Consciousness Program, Canadian Institute for Advanced Research (CIFAR), 661 University Avenue, Suite 505, Toronto, ON M5G 1M1 Canada}

\author{Andrei Khrennikov\,\orcidlink{0000-0002-9857-0938}}
\email{andrei.khrennikov@lnu.se}
\affiliation{Center for Mathematical Modeling in Physics and Cognitive Sciences, Linnaeus University, V{\"a}xj{\"o}, SE-351 95, Sweden}

\author{Kazunori Kondo\,\orcidlink{0009-0002-5798-2837}}
\email{kondo.hus@osaka-u.ac.jp}
\affiliation{Graduate School of Human Sciences, Department of Human Sciences, The University of Osaka, 1-2 Yamadaoka, Suita, Osaka, 565-0871, Japan}

\date{\today}

\begin{abstract}
Quantum logic is usually presented as a non-classical departure from ordinary reasoning forced on us by quantum mechanics, with classical logic kept as the secure starting point.
We argue for the opposite order of explanation in a finite and fully computable setting.
The free orthomodular lattice on two generators has ninety-six elements, the direct product of a six-element non-distributive factor and a sixteen-element Boolean factor.
Reading the first factor as a register of contexts and the second as Boolean content, we obtain a calculus whose elements are context--bit-vector pairs and whose operations act component by component.
With this calculus we establish three results.
First, we classify the six layers by commutativity, identifying the central kernel of context-neutral propositions together with a dual central layer in which all complementary contexts are present.
Second, we show that orthocomplementation rearranges the layers exactly as the complementation of the small factor rearranges its elements, which makes the duality among the layers rigid rather than accidental.
Third, we prove that the operation forgetting the context is a surjective homomorphism of orthocomplemented lattices whose quotient is the classical Boolean algebra, so that classical logic is a six-to-one, information-losing image of the contextual calculus.
\end{abstract}

\maketitle

\section{Introduction}
\label{sec:intro}
Quantum logic was introduced by Birkhoff and von Neumann~\cite{birkhoff1936} as the calculus of experimental propositions about a quantum system, with the lattice of projections of a Hilbert space taking the place of a Boolean algebra and with the distributive law no longer valid.
Since then it has most often been read as a non-classical departure from ordinary reasoning, one made necessary by the peculiarities of quantum mechanics, while classical Boolean logic is kept as the secure foundation against which such departures are measured.
The present paper argues, in a finite and completely computable setting, that this order of explanation can be reversed.
Classical logic is better understood as a quotient of a contextual quantum logic, obtained by forgetting which context a proposition belongs to.

Many cognitive and psychological experiments share a feature that ordinary logic does not record, namely that an answer depends on the question asked, on the order of questions, and on the context in which a question is posed.
Classical Boolean logic can store the final answer, but it does not store the context-dependent process by which that answer became determinate.
The order of questions, the task demand, and the measurement basis selected by an experiment jointly shape which proposition becomes determinate, which is the reason classical Boolean logic is insufficient as a generative logic even though it remains adequate as a record of final answers.
The present paper isolates, in a finite and fully computable model, the logical operation that turns such a contextual process into a classical Boolean record, and it does so as one part of a three-paper program on measurement strength~\cite{emori2026cc,emori2026pi}.
Ref.~\cite{emori2026cc} places major theories of consciousness and their experimental protocols along a weak-to-strong measurement-strength axis, Ref.~\cite{emori2026pi} supplies the dynamics of candidate histories, imaginary-time rectification, and the system--probe coupling that implements that axis, and the present paper supplies the logical machinery, a finite context calculus and the context-forgetting projection by which a contextual process becomes a classical Boolean record.

Several independent lines of work indicate that contextuality is intrinsic to any system with incompatible observables, rather than being an anomaly to be explained away.
Bohr's principle of complementarity~\cite{bohr1928,bohr1937, bohr1958} already held that quantum observables have no fixed meaning independent of an experimental context, so that incompatible observables demand mutually exclusive measurement arrangements.
The Kochen--Specker theorem~\cite{kochen1967} made this precise by excluding any context-independent valuation that respects the functional relations among observables, and the algebraic structure of contextual hidden-variable theories was later characterized through quantum set theory by Ozawa~\cite{ozawa2022hv}.
On the side of mathematical logic, Takeuti~\cite{takeuti1981a} founded quantum set theory on the orthomodular lattice of projections, and Ozawa developed it into a full semantics~\cite{ozawa2016,ozawa2017,ozawa2021a,ozawa2021b,ozawa2026} in which a transfer principle bounds the quantum truth value of every theorem of Zermelo--Fraenkel set theory with choice from below by a commutator, an element of the lattice that measures the degree of commutativity of the constants that appear~\cite{ozawa2021a}.
The long-standing arbitrariness in the choice of a quantum conditional was resolved by Hardegree~\cite{hardegree1981}, who isolated three polynomially definable conditionals, and was organized by Ozawa~\cite{ozawa2021a,ozawa2026} into a six-member family that is again governed by the commutator.
In a complementary direction, sheaf- and topos-theoretic accounts~\cite{abramsky2011,isham1998,doring2008} treat contextuality as the failure of consistent local data on the commutative subalgebras to glue into a single global assignment, and the orthomodular structure carried by a pair of propositions has recently been generated categorically, as a left adjoint, by Gunji \textit{et al.}~\cite{gunji2026a}.

These developments share an algebraic core.
By the commutator decomposition theorem~\cite{chevalier1989,ozawa2016,ozawa2026}, the sublogic generated by a subset of an orthomodular lattice is the direct product of a Boolean factor, on which the generators commute, and a residual factor that carries no Boolean part.
For two free generators this core is the free orthomodular lattice on two generators, which has exactly ninety-six elements in the enumeration of Beran~\cite{beran1985} and was confirmed to have no further two-variable operations by later computation~\cite{megill2001,hycko2005}.
In this case the decomposition is the direct product of a six-element non-Boolean factor, the modular ortholattice $\MO_2$ known as the Chinese lantern, and the sixteen-element Boolean algebra $2^4$.
What has not been isolated and studied in its own right, although it is implicit in the work just cited, is a triple of objects.
The first is an explicit calculus that writes each of the ninety-six elements as a pair consisting of a context drawn from the small factor and a Boolean bit vector, and computes every lattice operation component by component.
The second is a classification of the resulting six layers by their commutativity.
The third is the operation that forgets the context factor, together with the precise sense in which it turns the slogan ``classical logic is a quotient of quantum logic'' into a theorem.

This paper supplies these three objects and draws out their consequences.
We introduce the context--bit-vector calculus and prove, by reducing commutativity in the product to commutativity in the small factor, that the central elements of the lattice are exactly those of two of the six layers, namely the Boolean kernel of context-neutral propositions and a second, dual central layer in which the four complementary contexts are simultaneously present (Proposition~\ref{prop:commute}).
We then show that orthocomplementation permutes the six layers in exactly the way the orthocomplementation of the small factor permutes its six elements, so that the vertical and horizontal dualities among the layers are forced by that factor rather than being contingent (Proposition~\ref{prop:rigidity}).
Our main result, in Section~\ref{sec:projection}, is that the context-forgetting projection onto the Boolean factor is a surjective homomorphism of orthocomplemented lattices whose kernel congruence identifies precisely those elements that share a bit vector, so that the classical sixteen-element Boolean algebra is the quotient of the ninety-six element lattice by this congruence (Proposition~\ref{prop:boolean-quotient}).
Because every fiber of the projection has six elements, classical logic emerges as a uniform six-to-one, information-losing image of the contextual calculus.
We then explain how the projection acquires physical content once one adjoins the Born statistical formula and represents propositions by projections, which is the point at which it makes contact with Bohr's reading of complementarity and with the measurement-strength account of cognition~\cite{emori2026cc,emori2026pi}.
The reading that emerges reframes three habitual views.
It presents (i) classical logic as a quotient rather than quantum logic as an exception, (ii) the classical regime as a limit of information loss rather than as a foundation, and (iii) the failure of distributivity not as a defect but as a diagnostic for the presence of inter-context structure.
This perspective is consistent with the characterization of contextual hidden-variable theories given in quantum set theory~\cite{ozawa2022hv}.

We add two clarifications that hold throughout the paper.
First, the parallels we draw with the physical vocabulary of superposition, collapse, and erasure are structural analogies internal to the logic and are not claims about the dynamics of a Hilbert-space description, since physical content enters only where we explicitly adjoin the Born statistical formula.
Second, where we speak of an intermediate ``intuitionistic'' regime in Takeuti's sense we do so heuristically, because the natural candidate cannot be realized as a lattice quotient, as we show in Section~\ref{sec:discussion}.
The aim of this model is not to attach quantum vocabulary to cognition.
It is to track information that classical records discard, namely which context a proposition belongs to, whether two propositions can be combined without changing one another, and what is lost when a context-dependent process is forced into a context-free answer.
The enumeration of all ninety-six elements in Appendix~\ref{app:enumeration}, together with their formulas and vectors, has been verified by direct computation in the model.

The remainder of the paper is organized as follows.
Section~\ref{sec:prelim} recalls orthomodular lattices, the commutator and the decomposition theorem, and the quantum conditionals, and fixes the interface to quantum set theory.
Section~\ref{sec:stratification} develops the computational core of the paper, namely the context--bit-vector calculus, the classification of the six layers by commutativity, and the dynamics of context under meet and join.
Section~\ref{sec:dual-symmetry} establishes the layer dualities and proves their rigidity.
Section~\ref{sec:projection} contains the main result, the context-forgetting projection and its physical reading.
Section~\ref{sec:discussion} discusses the relation to Takeuti's three logics, to measurement strength, and to sheaf- and topos-theoretic approaches, and connects the forgetting projection to the strong-measurement limit of Refs.~\cite{emori2026cc,emori2026pi}.
Section~\ref{sec:conclusion} concludes.
Two appendices give the full enumeration and a set of worked computations.

\section{Mathematical Preliminaries}
\label{sec:prelim}
\subsection{Orthomodular lattices, commutators, and the decomposition theorem}
\label{sec:coml}
A \emph{complete orthomodular lattice} is a structure $\mathcal{L} = (L, \leq, {}^\perp, 0, 1)$ in which $(L, \leq)$ is a complete lattice, so that every subset $S \subseteq L$ has a supremum $\bigvee S$ and an infimum $\bigwedge S$ with $0 = \bigwedge L$ and $1 = \bigvee L$, and ${}^\perp$ is an orthocomplementation, a unary operation that satisfies $P \leq Q \Rightarrow Q^\perp \leq P^\perp$, $(P^\perp)^\perp = P$, $P \vee P^\perp = 1$, and $P \wedge P^\perp = 0$.
The orthomodular law requires that $P \leq Q$ imply $Q = P \vee (P^\perp \wedge Q)$~\cite{husimi1937}.
Following Ozawa~\cite{ozawa2021a,ozawa2026} we call a complete orthomodular lattice a \emph{logic} and identify orthomodular structure with the algebraic basis of logical operations.

Two elements $P, Q$ of a logic \emph{commute}, written $P \commutes Q$, when $P = (P \wedge Q) \vee (P \wedge Q^\perp)$.
This expresses the orthogonal decomposability of $P$ relative to $Q$, and a logic is a Boolean algebra exactly when all of its elements pairwise commute.
For a subset $A \subseteq L$ the commutant $A^{!} = \{P \in L \mid P \commutes Q \text{ for all } Q \in A\}$ is a complete subalgebra, and the central commutator
\eq{
  \com(A) = \max\{E \in A^{!} \mid P \wedge E \commutes Q \wedge E \text{ for all } P, Q \in A\}
  \label{eq:com-def}
}
is the largest element below which all of $A$ becomes mutually commutative~\cite{chevalier1989,ozawa2016,ozawa2026}.
For a two-element set this central commutator has the closed form
\eq{
  \com(P,Q) = (P \wedge Q) \vee (P \wedge Q^\perp) \vee (P^\perp \wedge Q) \vee (P^\perp \wedge Q^\perp),
  \label{eq:com-two}
}
a theorem of Ozawa~\cite[Theorem~2.1(i)]{ozawa2026} rather than a separate definition.
The same join of four meets is computed in Appendix~\ref{app:examples}, where it is read directly as the central commutator of Eq.~\eqref{eq:com-def} and not as a distinct, Marsden-style construction.
\begin{theorem}[Decomposition theorem~\cite{chevalier1989,ozawa2016,ozawa2026}]
\label{thm:decomp}
Let $A$ be a subset of a logic $\mathcal{L}$. Then the sublogic $A^{!!}$ generated by $A$ is isomorphic to the direct product
\eq{
  A^{!!} \cong [0, \com(A)]_{A^{!!}} \times [0, \com(A)^\perp]_{A^{!!}}, \nonumber
}
in which the first factor is a complete Boolean algebra and the second is a complete orthomodular lattice with no non-trivial Boolean factor.
\end{theorem}
Here $[0,E]_{A^{!!}} \coloneqq \{P \in A^{!!} \mid 0 \le P \le E\}$.
Classical sublogics arise where the commutator is maximal, and genuinely quantum structure survives in its orthogonal complement.
We apply this to the two free generators $x$ and $y$.
The lattice they generate is the free orthomodular lattice on two generators, denoted $F_2$, which has ninety-six elements~\cite{beran1985}.
Theorem~\ref{thm:decomp} with $A=\{x,y\}$ identifies the Boolean factor $[0,\com(x,y)]$ with the free Boolean algebra on two generators, namely $2^4$, and the residual factor $[0,\com(x,y)^\perp]$ with the six-element modular ortholattice $\MO_2$.
The decomposition used throughout is therefore
\eq{
  F_2 \cong \MO_2 \times 2^4. \label{eq:decomp-f2}
}
The ninety-six element count and the enumeration are due to Beran~\cite{beran1985}, and the identification of the two factors as $\MO_2$ and $2^4$ follows from Theorem~\ref{thm:decomp} together with that enumeration, not from either ingredient alone.
Quantum set theory~\cite{ozawa2016,ozawa2017,ozawa2021a} provides a semantics for orthomodular-valued truth in which the commutator controls the extent of classical behavior, and we return to that interface, through the projection of Section~\ref{sec:projection}, in Section~\ref{sec:discussion}.

\subsection{The quantum conditionals}
\label{sec:conditionals}
A Boolean algebra has a single material conditional $P \to Q = \neg P \vee Q$, but an orthomodular lattice has no canonical analogue.
Hardegree~\cite{hardegree1981} fixed the choice by minimal requirements, which we state precisely.
An \emph{ortholattice polynomial} in the variables $P,Q$ is a term built from $P$, $Q$ by the operations $\wedge$, $\vee$, ${}^\perp$.
\begin{definition}[Quantum material conditional~\cite{hardegree1981}]
\label{def:qmc}
A \emph{quantum material conditional} on a logic $\mathcal{L}$ is a binary operation $\Rightarrow$ for which there is an ortholattice polynomial $p(P,Q)$ with $P \Rightarrow Q = p(P,Q)$ for all $P, Q \in L$, and which satisfies, for all $P, Q \in L$, the three minimal implicative conditions
\begin{enumerate}[label=\textup{(\arabic*)},leftmargin=2.4em]
  \item[\textup{(E)}] $P \Rightarrow Q = 1$ if and only if $P \leq Q$;
  \item[\textup{(MP)}] $P \wedge (P \Rightarrow Q) \leq Q$ \quad (modus ponens);
  \item[\textup{(MT)}] $Q^\perp \wedge (P \Rightarrow Q) \leq P^\perp$ \quad (modus tollens).
\end{enumerate}
\end{definition}
\begin{theorem}[Hardegree~\cite{hardegree1981}; Ozawa~\cite{ozawa2021a,ozawa2026}]
\label{thm:hardegree}
A logic admits exactly three quantum material conditionals in the sense of Definition~\ref{def:qmc}, and they are the Sasaki conditional, the contrapositive Sasaki conditional, and the relevance conditional,
\eq{
  P \impS Q &= P^\perp \vee (P \wedge Q), \nonumber\\
  P \impC Q &= (P \vee Q)^\perp \vee Q, \nonumber\\
  P \impR Q &= (P \wedge Q) \vee (P^\perp \wedge Q) \vee (P^\perp \wedge Q^\perp). \nonumber
}
Every binary operation given by an ortholattice polynomial and satisfying \textup{(E)}, \textup{(MP)}, and \textup{(MT)} equals one of these three, and conversely each of the three satisfies the conditions.
\end{theorem}
The three conditionals agree with the classical conditional $P^\perp \vee Q$ whenever $P \commutes Q$, and they differ only through the commutator.
Ozawa~\cite{ozawa2021a,ozawa2026} made the dependence explicit in the identities
\eq{
  P \impS Q &= (P \impR Q) \vee \bigl(P^\perp \wedge \com(P,Q)^\perp\bigr), \nonumber\\
  P \impC Q &= (P \impR Q) \vee \bigl(Q \wedge \com(P,Q)^\perp\bigr), \nonumber
}
which exhibit the Sasaki and contrapositive Sasaki conditionals as the relevance conditional augmented on the non-commutative part $\com(P,Q)^\perp$.
The relaxation of (E) to the local Boolean condition, that $P \leq Q$ imply $P \Rightarrow Q = P^\perp \vee Q$, admits a slightly larger family, classified as follows.
\begin{theorem}[Kotas--Ozawa~\cite{kotas1967,ozawa2021a,ozawa2026}]
\label{thm:kotas}
A logic admits exactly six binary operations given by ortholattice polynomials and satisfying the local Boolean condition. They are the operations $P \Rightarrow_j Q = \mathrm{b}_n(P,Q) \vee \alpha_j$ for $j = 0,\dots,5$, where $\mathrm{b}_n(P,Q) = (P \wedge Q) \vee (P^\perp \wedge Q) \vee (P^\perp \wedge Q^\perp)$ is the disjunctive normal form of the classical conditional and $\alpha_j$ ranges over the six elements
\eq{
  0, \quad P \wedge \com(P,Q)^\perp, \quad Q \wedge \com(P,Q)^\perp, \quad P^\perp \wedge \com(P,Q)^\perp, \quad Q^\perp \wedge \com(P,Q)^\perp, \quad \com(P,Q)^\perp. \nonumber
}
\end{theorem}
The members $j = 0, 2, 3$ are the relevance, contrapositive Sasaki, and Sasaki conditionals, and the member $j = 5$ is the classical conditional $P^\perp \vee Q$, which satisfies the local Boolean condition but violates (E) in a non-Boolean lattice.
The two remaining members, written out for the generators, are
\eq{
  x \Rightarrow_1 y &= (x^\perp \vee y) \wedge \bigl(x \vee (x^\perp \wedge y) \vee (x^\perp \wedge y^\perp)\bigr), \nonumber\\
  x \Rightarrow_4 y &= (x^\perp \vee y) \wedge \bigl(y^\perp \vee (x^\perp \wedge y) \vee (x \wedge y)\bigr). \nonumber
}
Theorem~\ref{thm:kotas} also fixes the relation between the ninety-six two-variable operations and the sixteen classical truth functions.
Writing $\beta(P,Q)$ for the disjunctive normal form of the two-variable Boolean function that an operation realizes, every such operation has the normal form $\beta(P,Q) \vee (\varepsilon \wedge \com(P,Q)^\perp)$, with $\beta$ ranging over the sixteen Boolean parts and $\varepsilon \in \{0, P, P^\perp, Q, Q^\perp, 1\}$, giving $16 \times 6 = 96$ operations~\cite{kotas1967,ozawa2021a}.
The sixteen operations with $\varepsilon = 0$ coincide with their Boolean part $\beta$ everywhere, and the remaining eighty differ from it only on $\com(P,Q)^\perp$, so that they agree with their classical counterpart precisely when $x \commutes y$.

For each quantum material conditional the \emph{dual conjunction} is $P \otimes Q \coloneqq (P \Rightarrow Q^\perp)^\perp$, which keeps the duality between bounded universal and existential quantifiers in quantum set theory~\cite{ozawa2021a,ozawa2026}.
The three dual conjunctions are
\eq{
  P \dcS Q &= (P \wedge Q) \vee \bigl(P \wedge \com(P,Q)^\perp\bigr), \nonumber\\
  P \dcC Q &= (P \wedge Q) \vee \bigl(Q \wedge \com(P,Q)^\perp\bigr), \nonumber\\
  P \dcR Q &= (P \wedge Q) \vee \com(P,Q)^\perp. \label{eq:dual-conj}
}
For the generators one has $x \dcS y = x \wedge (x^\perp \vee y)$ and $x \dcC y = (x \vee y^\perp) \wedge y$.
The relevance dual conjunction is $x \dcR y = (x^\perp \vee y) \wedge (x \vee y^\perp) \wedge (x \vee y)$, which equals $(x \wedge y) \vee \com(x,y)^\perp$.
In the calculus of Section~\ref{sec:stratification} this element is $(1,1000)$, which is element No.~82 of the enumeration in Appendix~\ref{app:enumeration}, and it is distinct from $x \wedge y = (0,1000)$.
The bare equality $P \dcR Q = P \wedge Q$ holds only in the commutative case, where $\com(P,Q)^\perp = 0$, and the term $\com(P,Q)^\perp$ is present otherwise, as the normal form of Theorem~\ref{thm:kotas} requires.
All three dual conjunctions reduce to $P \wedge Q$ when $P \commutes Q$.

In the enumeration of Beran~\cite{beran1985} the Sasaki, contrapositive Sasaki, and relevance conditionals carry the numbers $78$, $46$, and $14$, with
\eq{
  x \impS y = x^\perp \vee (x \wedge y), \quad
  x \impC y = y \vee (y^\perp \wedge x^\perp), \quad
  x \impR y = (x \wedge y) \vee (x^\perp \wedge y) \vee (x^\perp \wedge y^\perp). \nonumber
}
We adopt Beran's numbering, and we confirm the position of these three operations, and of every other element, by direct computation in the calculus of Section~\ref{sec:stratification}, with the full check recorded in Appendix~\ref{app:enumeration}.

\section{The Context Calculus and the Six-Layer Stratification}
\label{sec:stratification}
This section develops the computational core of the paper.
We write each element of $F_2$ as a pair consisting of a context and a bit vector and record the resulting component-by-component operations in Definition~\ref{def:operations}.
We then classify the six layers of the lattice by commutativity in Proposition~\ref{prop:commute}, and we describe in Theorem~\ref{thm:context-dynamics} how meet and join move an element from one layer to another.
These facts are the input to the duality results of Section~\ref{sec:dual-symmetry} and to the projection of Section~\ref{sec:projection}.

By the decomposition \eqref{eq:decomp-f2}, the ninety-six elements of $F_2$ split into six layers of sixteen elements each, one layer for each of the six contexts of $\MO_2$.
We write the six contexts as $0,a,b,a^\perp,b^\perp,1$ and label the layers $B01,\dots,B06$ as in Table~\ref{tab:six-layer}.
The kernel $B01$ has context $0$, the four single-context layers $B02$--$B05$ have contexts $a,b,b^\perp,a^\perp$, and the layer $B06$ has context $1$.
The full enumeration is given in Appendix~\ref{app:enumeration}.

We keep one distinction in view throughout, because it is easy to lose.
The symbol $c$ always names a value of the context factor $\MO_2$, whereas an element of $F_2$ is the pair $(c,\bvec b)$.
In particular the top context $c=1$ of $\MO_2$, viewed as an element of $F_2$, is the commutator complement $\com(x,y)^\perp = (1,0000)$, which differs from the global top $1_{F_2} = (1,1111)$.
\begin{table}[htbp]
\caption{\label{tab:six-layer}The six-layer stratification of $F_2 \cong \MO_2 \times 2^4$. The context is the value of the $\MO_2$ factor; the elements of contexts $0$ and $1$ are central in $F_2$, and the elements of the four atomic contexts are non-central, by Proposition~\ref{prop:commute}. Each key example is identified by its Beran number.}
\centering
\begin{tabular}{c c c l}
\hline\hline
Layer & Beran range & Context $c$ & Key example \\
\hline
$B01$ & 01--16 & $0$ & No.~14: relevance conditional $x \impR y$ \\
$B02$ & 17--32 & $a = x \wedge \com(x,y)^\perp$ & No.~18: Sasaki conjunction $x \dcS y$ \\
$B03$ & 33--48 & $b = y \wedge \com(x,y)^\perp$ & No.~46: contrapositive Sasaki conditional $x \impC y$ \\
$B04$ & 49--64 & $b^\perp = y^\perp \wedge \com(x,y)^\perp$ & No.~58: generator complement $y^\perp$ \\
$B05$ & 65--80 & $a^\perp = x^\perp \wedge \com(x,y)^\perp$ & No.~78: Sasaki conditional $x \impS y$ \\
$B06$ & 81--96 & $1 = \com(x,y)^\perp$ & No.~94: classical conditional $x \to y$ \\
\hline\hline
\end{tabular}
\end{table}

For a reader meeting the stratification for the first time, the picture is as follows.
Each element carries two independent pieces of data, a context and a bit vector.
The bit vector is ordinary Boolean content, the value of a two-variable formula on the four Boolean cases.
The context records which experimental frame the proposition belongs to, with $0$ and $1$ the two context-neutral (central) values and $a,b,a^\perp,b^\perp$ the four single-generator frames.
The six layers of Table~\ref{tab:six-layer} are the six possible contexts, and everything that follows turns on how the operations move an element between them.

\subsection{The component-by-component calculus}
\label{sec:comp-rep}
Each element $P \in F_2$ has a unique representation $P = (c, \bvec b)$ with context $c \in \MO_2$ and bit vector $\bvec b = (b_1,b_2,b_3,b_4) \in 2^4$.
The four bits record the value of $P$ on the four atoms of the Boolean subalgebra generated by $x$ and $y$, taken in the order $(x \wedge y,\ x \wedge y^\perp,\ x^\perp \wedge y,\ x^\perp \wedge y^\perp)$.
The lattice operations act on the two factors independently.
\begin{definition}[Operations on $(c,\bvec b)$]
\label{def:operations}
For $P_1 = (c_1, \bvec b_1)$ and $P_2 = (c_2, \bvec b_2)$,
\eq{
  P_1 \wedge P_2 = (c_1 \wedge_{\MO_2} c_2,\ \bvec b_1 \wedge_{2^4} \bvec b_2), \quad
  P_1 \vee P_2 = (c_1 \vee_{\MO_2} c_2,\ \bvec b_1 \vee_{2^4} \bvec b_2), \quad
  P^\perp = (c^{\perp}_{\MO_2},\ \bvec b^{\perp}_{2^4}), \nonumber
}
where $\wedge_{\MO_2}, \vee_{\MO_2}, {}^{\perp}_{\MO_2}$ are the operations of $\MO_2$ and $\wedge_{2^4}, \vee_{2^4}, {}^{\perp}_{2^4}$ are the bitwise AND, OR, and NOT on four bits.
\end{definition}
Definition~\ref{def:operations} is the statement that the isomorphism \eqref{eq:decomp-f2} is an isomorphism of orthocomplemented lattices, so that all ninety-six elements can be enumerated and every operation can be computed in constant time.
Worked examples are collected in Appendix~\ref{app:examples}.
For example, for $P_1 = (a, 1100) \in B02$ and $P_2 = (b, 1010) \in B03$ the context meet is $a \wedge_{\MO_2} b = 0$, because $a$ and $b$ are distinct atoms of $\MO_2$, and the bitwise meet is $1100 \wedge 1010 = 1000$, so that $P_1 \wedge P_2 = (0, 1000)$ lies in $B01$.
Combining elements from two different non-central contexts can thus force a descent to the Boolean kernel.

\subsection{The Boolean kernel and the central superposition layer}
\label{sec:b01-b06}
The kernel $B01$, with context $0$, consists of the sixteen elements $(0,\bvec b)$, and these are exactly the two-variable Boolean operations on commuting generators.
Representative elements are listed in Table~\ref{tab:b01}.
The three quantum conditionals coincide on $B01$, as the next corollary records.
\begin{table}[htbp]
\caption{\label{tab:b01}Selected elements of the Boolean kernel $B01$. The last column names the two-variable Boolean operation that has the same bit vector.}
\centering
\begin{tabular}{c l c l}
\hline\hline
Beran No. & Formula & Vector & Boolean operation \\
\hline
01 & $x \wedge x^\perp$ & $(0,0000)$ & falsity \\
02 & $x \wedge y$ & $(0,1000)$ & conjunction \\
06 & $(x \wedge y) \vee (x \wedge y^\perp)$ & $(0,1100)$ & projection onto $x$ \\
14 & $x \impR y$ & $(0,1011)$ & material conditional \\
16 & $\textstyle\bigvee$ of the four atoms & $(0,1111)$ & truth (top of $B01$) \\
\hline\hline
\end{tabular}
\end{table}
\begin{corollary}[Classical agreement in $B01$]
\label{cor:b01-agreement}
For all $P, Q \in B01$ the three quantum conditionals coincide, and they equal the classical conditional, $P \impS Q = P \impC Q = P \impR Q = P^\perp \vee Q$.
\end{corollary}
\begin{proof}
By Proposition~\ref{prop:commute} every element of $B01$ is central, so any $P,Q \in B01$ commute, and for commuting elements the central commutator is the global top, $\com(P,Q) = 1_{F_2} = (1,1111)$.
One sees this directly from Eq.~\eqref{eq:com-two}.
Writing $P = (0,\bvec b)$ and $Q = (0,\bvec b')$, the meets $P \wedge Q$, $P \wedge Q^\perp$, and $P^\perp \wedge Q$ have context $0$, whereas the fourth meet $P^\perp \wedge Q^\perp$ has context $1$, because the orthocomplements $P^\perp$ and $Q^\perp$ lie in $B06$ and have context $1$.
The four bit vectors of these meets are $\bvec b \wedge \bvec b'$, $\bvec b \wedge \neg\bvec b'$, $\neg\bvec b \wedge \bvec b'$, and $\neg\bvec b \wedge \neg\bvec b'$, whose join is $1111$.
The join of the four meets is therefore $(0 \vee 0 \vee 0 \vee 1,\ 1111) = (1,1111)$, an absolute identity in $F_2$ and not a relative one inside $B01$.
Hence $\com(P,Q)^\perp = 0_{F_2} = (0,0000)$, every augmentation term $\alpha_j = \varepsilon \wedge \com(P,Q)^\perp$ of Theorem~\ref{thm:kotas} vanishes, and all six conditionals reduce to $\mathrm{b}_n(P,Q) = P^\perp \vee Q$.
\end{proof}
Within $B01$ classical reasoning is internally stable, since the conditionals agree and the bit vector obeys the ordinary two-variable truth-functional calculus.
The kernel is not globally complete inside $F_2$, however, because the global top $1_{F_2}$ has context $1$ and so lies in $B06$ (Beran No.~96).
Reaching the whole lattice forces one to leave the kernel.

The layer $B06$, with context $1 = \com(x,y)^\perp$, is the second central layer.
Its elements do not fix a single generator-context but lie in the region where the complementary contexts are jointly present.
The term ``superposition'' here is internal to the lattice and algebraic.
It records that an element of $B06$ has a nontrivial image in each of the complementary single contexts, in the sense made precise after Theorem~\ref{thm:vertical}, and it is not a claim about physical quantum superposition, which enters only with the Born statistical formula of Section~\ref{sec:physical}.
Table~\ref{tab:b06} lists representative elements, including the global top and the commutator complement.
The next proposition records a feature that motivates the term ``superposition layer.''
The two central layers are insulated from the single-context layers, since the center $B01 \cup B06$ is closed under all lattice operations, so that no context-neutral datum can synthesize a specific experimental context.
\begin{table}[htbp]
\caption{\label{tab:b06}Selected elements of the central superposition layer $B06$, including the global top $1_{F_2}$ and the commutator complement $\com(x,y)^\perp$. The last column names the two-variable Boolean operation with the same bit vector. The $B06$ element with a given bit vector is distinct from the $B01$ element that carries the same Boolean operation, since the two differ in their context; in particular $\com(x,y)^\perp = (1,0000)$ is the bottom of $B06$, not the global bottom $0_{F_2} = (0,0000)$, which lies in $B01$.}
\centering
\begin{tabular}{c l c l}
\hline\hline
Beran No. & Formula & Vector & Boolean operation \\
\hline
81 & $(x \vee y) \wedge (x \vee y^\perp) \wedge (x^\perp \vee y^\perp) \wedge (x^\perp \vee y)$ & $(1,0000)$ & falsity (bit vector); $\com(x,y)^\perp$, bottom of $B06$ \\
94 & $x^\perp \vee y$ & $(1,1011)$ & material conditional \\
96 & $1$ & $(1,1111)$ & truth ($={}$global top $1_{F_2}$) \\
\hline\hline
\end{tabular}
\end{table}
\begin{proposition}[Closure of the central layers]
\label{prop:b06-classical}
The union $B01 \cup B06$ of the two central layers is a subalgebra of $F_2$, closed under $\wedge$, $\vee$, and ${}^\perp$. Consequently no finite sequence of lattice operations applied to elements of $B01 \cup B06$ can produce an element of the non-central single-context layers $B02$--$B05$.
\end{proposition}
\begin{proof}
By Proposition~\ref{prop:commute} the elements of $B01 \cup B06$ are exactly those whose context lies in $\{0,1\} \subseteq \MO_2$.
The subset $\{0,1\}$ is a subalgebra of $\MO_2$, since it is closed under ${}^\perp$, because $0^\perp = 1$ and $1^\perp = 0$, and under $\wedge,\vee$, because $0$ and $1$ are the bottom and top of $\MO_2$.
By the component-by-component operations of Definition~\ref{def:operations}, the context of any meet, join, or orthocomplement of elements with contexts in $\{0,1\}$ again lies in $\{0,1\}$.
Hence $B01 \cup B06 = \{0,1\} \times 2^4$ is closed under all lattice operations, and since the contexts $a,b,a^\perp,b^\perp$ of $B02$--$B05$ are not in $\{0,1\}$, no element of those layers is reachable.
\end{proof}
The point is again that the context, not the bit vector, controls reachability, since the non-central layers $B02$--$B05$ cannot be entered from the center, even though every bit vector already occurs there.
We will see in Section~\ref{sec:dual-symmetry} that $B06$ is forced into existence as the orthocomplement of $B01$, so the central superposition layer is not an optional addition to the lattice.

\subsection{Commutativity of the layers}
\label{sec:b02-05}
The four single-context layers $B02$--$B05$, with contexts $a,b,b^\perp,a^\perp$, are the only layers whose elements are non-central.
Representative elements appear in Table~\ref{tab:b02-05}.
The following proposition gives their commutativity in full, and it replaces the imprecise description, which one finds in informal accounts, that an element of one of these layers ``fails to commute with exactly one of $x,y,x^\perp,y^\perp$.''
That description is ill-posed, because commutation is invariant under orthocomplementation, so an element either commutes with both of $x$ and $x^\perp$ or with neither.
\begin{table}[htbp]
\caption{\label{tab:b02-05}Representative elements of the single-context layers $B02$--$B05$.}
\centering
\begin{tabular}{c c l c l}
\hline\hline
Layer & Beran No. & Formula & Vector & Description \\
\hline
$B02$ & 19 & $x \wedge (x^\perp \vee y^\perp)$ & $(a,0100)$ & no named operation \\
$B02$ & 22 & $x$ & $(a,1100)$ & generator $x$ \\
$B03$ & 46 & $y \vee (y^\perp \wedge x^\perp)$ & $(b,1011)$ & contrapositive Sasaki conditional \\
$B05$ & 75 & $x^\perp$ & $(a^\perp,0011)$ & generator complement $x^\perp$ \\
$B05$ & 78 & $x^\perp \vee (x \wedge y)$ & $(a^\perp,1011)$ & Sasaki conditional \\
\hline\hline
\end{tabular}
\end{table}
\begin{proposition}[Commutativity of the layers]
\label{prop:commute}
Let $P = (c,\bvec b) \in F_2$. Then $P$ is central, that is, $P$ commutes with every element of $F_2$, if and only if $c \in \{0,1\}$, equivalently $P \in B01 \cup B06$. Among the single-context layers, every element of $B02$ and of $B05$ commutes with $x$ and with $x^\perp$ but with neither $y$ nor $y^\perp$, and every element of $B03$ and of $B04$ commutes with $y$ and with $y^\perp$ but with neither $x$ nor $x^\perp$.
\end{proposition}
\begin{proof}
We first reduce commutativity in $F_2$ to commutativity in $\MO_2$.
In a direct product of orthomodular lattices, $(c_1,\bvec b_1) \commutes (c_2,\bvec b_2)$ holds if and only if $c_1 \commutes c_2$ in $\MO_2$ and $\bvec b_1 \commutes \bvec b_2$ in $2^4$, because the defining identity $P = (P\wedge Q)\vee(P\wedge Q^\perp)$ holds in a product exactly when it holds in each factor.
The factor $2^4$ is Boolean, so $\bvec b_1 \commutes \bvec b_2$ always holds, and commutativity in $F_2$ is governed entirely by the contexts.
We next determine commutativity in $\MO_2$.
Let $\alpha$ be an atom of $\MO_2$ and let $c \in \MO_2$.
If $c \in \{0,\alpha,\alpha^\perp,1\}$ then $c$ and $\alpha$ lie in a common Boolean subalgebra and so commute.
If $c$ is an atom different from $\alpha$ and from $\alpha^\perp$, then $c \wedge \alpha = c \wedge \alpha^\perp = 0$, because distinct atoms of $\MO_2$ have meet $0$, whence $(c \wedge \alpha)\vee(c \wedge \alpha^\perp) = 0 \neq c$ and $c$ does not commute with $\alpha$.
Therefore $c$ commutes with $\alpha$ if and only if $c \in \{0,\alpha,\alpha^\perp,1\}$.
We now apply this.
The generators have contexts $a$ and $b$, that is, $x$ has $\MO_2$-component $a$ and $y$ has $\MO_2$-component $b$.
By the previous paragraph, and using that commutation is invariant under orthocomplementation, an element of context $c$ commutes with $x$ and $x^\perp$ exactly when $c \in \{0,a,a^\perp,1\}$, and with $y$ and $y^\perp$ exactly when $c \in \{0,b,b^\perp,1\}$.
If $c \in \{0,1\}$ then both hold, and $P$ commutes with the generators, hence with the sublogic they generate, which is all of $F_2$, so $P$ is central.
If $c \in \{a,a^\perp\}$ then $P$ commutes with $x,x^\perp$ but not with $y,y^\perp$, and if $c \in \{b,b^\perp\}$ the reverse holds.
Since $B02$ and $B05$ carry contexts $a,a^\perp$ and $B03$ and $B04$ carry contexts $b,b^\perp$, the stated profiles follow.
\end{proof}
\begin{corollary}[Context-dependent conditionals]
\label{cor:context-dependent}
For $P \in B02$ and $Q \in B03$ the three quantum conditionals are in general distinct, $P \impS Q \neq P \impC Q \neq P \impR Q$.
\end{corollary}
\begin{proof}
By Proposition~\ref{prop:commute} an element of $B02$ does not commute with $y$, and an element of $B03$ has context $b$, so $\com(P,Q) < 1$ and $\com(P,Q)^\perp \neq 0$.
The augmentation terms $\alpha_j$ of Theorem~\ref{thm:kotas} are then non-zero, and the values $P^\perp \wedge \com(P,Q)^\perp$ and $Q \wedge \com(P,Q)^\perp$ that distinguish the Sasaki and contrapositive Sasaki conditionals are in general different, which separates the three conditionals.
\end{proof}
The asymmetry is already visible on the generators.
The Sasaki conditional $x \impS y$ is No.~78 and lies in $B05$, with context $a^\perp$, while the contrapositive Sasaki conditional $x \impC y$ is No.~46 and lies in $B03$, with context $b$.
The Sasaki conditional therefore takes the context of the antecedent and the contrapositive Sasaki conditional takes the context of the consequent, so the order of antecedent and consequent already affects the result, before any temporal dynamics is introduced.
Each single-context layer represents reasoning under a fixed context, in which the bit vector still computes classically but the element fails to commute with one of the generators.
Once propositions from two different single contexts are combined, a meet or a join can move the result out of this regime, often down to $B01$, which is one concrete way to see why implication-like operations are not globally uniform in an orthomodular lattice.

\subsection{Context dynamics under meet and join}
\label{sec:context-dynamics}
The contexts carry the lattice order of $\MO_2$, in which $0$ is the bottom, $1$ is the top, and the four atoms are pairwise incomparable.
A meet or join of two elements moves the result to the layer of the context meet or context join.
\begin{theorem}[Context dynamics]
\label{thm:context-dynamics}
Let $P_1 = (c_1,\bvec b_1)$ and $P_2 = (c_2,\bvec b_2)$. Then the context of $P_1 \wedge P_2$ is $c_1 \wedge_{\MO_2} c_2$ and the context of $P_1 \vee P_2$ is $c_1 \vee_{\MO_2} c_2$. Consequently $P_1 \wedge P_2 \in B01$ if and only if $c_1 \wedge_{\MO_2} c_2 = 0$, and $P_1 \vee P_2 \in B06$ if and only if $c_1 \vee_{\MO_2} c_2 = 1$. In particular, if $c_1$ and $c_2$ are distinct atoms, then $P_1 \wedge P_2 \in B01$ and $P_1 \vee P_2 \in B06$.
\end{theorem}
\begin{proof}
The statement about the contexts of $P_1 \wedge P_2$ and $P_1 \vee P_2$ is the first component of Definition~\ref{def:operations}.
Since the layers are indexed by the context, $P_1 \wedge P_2$ lies in $B01$ exactly when its context is $0$ and in $B06$ exactly when its context is $1$, which gives the two equivalences.
If $c_1$ and $c_2$ are distinct atoms of $\MO_2$, then $c_1 \wedge_{\MO_2} c_2 = 0$ and $c_1 \vee_{\MO_2} c_2 = 1$, because distinct atoms of $\MO_2$ have meet $0$ and join $1$, and the last claim follows.
\end{proof}
The three transitions $(a,1100) \wedge (b,1010) = (0,1000)$, $(0,0110) \vee (a,1000) = (a,1110)$, and $(a,1000) \vee (b,0100) = (1,1100)$ realize $B02 \wedge B03 \to B01$, $B01 \vee B02 \to B02$, and $B02 \vee B03 \to B06$.
The third is a forced ascent to $B06$.
Combining two elements from distinct atomic contexts produces a central, context-superpositional element, even when the two bit vectors remain Boolean.
Meet therefore acts as a contextual intersection, computing the common ground of two perspectives and possibly collapsing to the neutral context, while join acts as a contextual union, combining perspectives and possibly forcing the ascent to the central layer.
The combination is not reversible at the level of context, since a join can recover a higher context but cannot record which contexts produced it.
Both $a \vee b$ and $a^\perp \vee b^\perp$ equal $\com(x,y)^\perp$, so the data of the two contributing contexts are lost.

In cognitive terms, a meet asks for what two contexts have in common, and a join combines two frames into a more encompassing context.
This is why the order and pairing of questions can matter, since the same Boolean content may be carried by different context histories, and those histories are exactly what the projection of Section~\ref{sec:projection} discards.
The conclusion is that the operations of $F_2$ update the context together with the Boolean content, so that meet and join are operators of contextual restriction and contextual combination.
The information loss that the projection of Section~\ref{sec:projection} makes precise is already present here, at the level of these layer transitions.

\section{Layer Dualities and Their Rigidity}
\label{sec:dual-symmetry}
Orthocomplementation is an involution that reverses the order, since $(P^\perp)^\perp = P$ and $P \leq Q$ implies $Q^\perp \leq P^\perp$, and it sends each element to a complement, since $P \vee P^\perp = 1$ and $P \wedge P^\perp = 0$.
In $F_2$ it acts on the six layers, pairing $B01$ with $B06$ and exchanging $B02$ with $B05$ and $B03$ with $B04$.
We establish the two dualities and then prove that they are exactly the orthocomplementation of the $\MO_2$ factor, which is what makes the layer picture rigid.

\subsection{Vertical and horizontal dualities}
\label{sec:vertical-horizontal}
\begin{theorem}[Vertical duality]
\label{thm:vertical}
Orthocomplementation restricts to an order-reversing bijection between $B01$ and $B06$.
\end{theorem}
\begin{proof}
Let $P = (0,\bvec b) \in B01$.
By Definition~\ref{def:operations}, $P^\perp = (0^\perp_{\MO_2}, \neg\bvec b) = (1,\neg\bvec b)$, which lies in $B06$ because $0^\perp = 1$ in $\MO_2$.
Conversely $(1,\bvec b)^\perp = (0,\neg\bvec b) \in B01$, so the restriction of ${}^\perp$ to $B01$ is a bijection onto $B06$ with inverse the restriction to $B06$.
It reverses the order, because for $P_1 = (0,\bvec b)$ and $P_2 = (0,\bvec b')$ one has $P_1 \leq P_2 \iff \bvec b \leq \bvec b' \iff \neg\bvec b' \leq \neg\bvec b \iff P_2^\perp \leq P_1^\perp$, where the middle equivalence is the order reversal of the Boolean complement on $2^4$.
\end{proof}
Theorem~\ref{thm:vertical} shows that the central superposition layer is not an optional addition.
Once $B01$ is present, its image under orthocomplementation is a sixteen-element layer of central elements at context $1$, namely $B06$, so the orthomodular structure forces $B06$ into existence.
The name superposition layer is justified at the level of the operations.
For $P \in B06$ and each single-context layer $B0i$ with $i \in \{2,3,4,5\}$ there is an element $Q_i \in B0i$ with $P \wedge Q_i \in B0i$, because the context meet $1 \wedge_{\MO_2} c_i = c_i$ keeps the result in $B0i$.
Every element of $B06$ other than its minimum $\com(x,y)^\perp=(1,0000)$ therefore has a non-trivial image in each of the four complementary single contexts, which is the sense in which its content is distributed over the contexts rather than localized in any one of them.
\begin{theorem}[Horizontal duality]
\label{thm:horizontal}
Orthocomplementation restricts to order-reversing bijections $B02 \leftrightarrow B05$ and $B03 \leftrightarrow B04$.
\end{theorem}
\begin{proof}
Let $P = (a,\bvec b) \in B02$, where $a = x \wedge \com(x,y)^\perp$.
The orthocomplement of $a$ in $\MO_2$ is $a^\perp = x^\perp \wedge \com(x,y)^\perp$, so by Definition~\ref{def:operations} we have $P^\perp = (a^\perp, \neg\bvec b) \in B05$, and $(a^\perp)^\perp = a$ gives the inverse, exactly as in the proof of Theorem~\ref{thm:vertical}.
For instance No.~22, the generator $x = (a,1100)$, maps to No.~75, the complement $x^\perp = (a^\perp,0011)$, and No.~19, the element $(a,0100)$, maps to No.~78, the Sasaki conditional $(a^\perp,1011)$.
The argument for $B03 \leftrightarrow B04$ is identical with $b$ in place of $a$.
Order reversal follows as before.
\end{proof}
Horizontal duality identifies the act of fixing a context with the act of fixing its complement, at the level of $\MO_2$, through $a \leftrightarrow a^\perp$ and $b \leftrightarrow b^\perp$.
Orthocomplementation does more than negate the Boolean content, since it also moves a proposition between the two complementary single-context perspectives.

\subsection{Functoriality and rigidity}
\label{sec:functorial}
\begin{definition}[The order category and its duality]
\label{def:functor}
Regard $F_2$ as a category $\mathcal{C}_{F_2}$ whose objects are the elements of $F_2$ and in which there is a unique morphism $P \to Q$ exactly when $P \leq Q$. A contravariant endofunctor of $\mathcal{C}_{F_2}$ is then a map on objects that turns each morphism $P \to Q$ into a morphism $Q^\perp \to P^\perp$.
\end{definition}
\begin{theorem}[Functoriality of orthocomplementation]
\label{thm:functoriality}
The map ${}^\perp$ is a contravariant endofunctor of $\mathcal{C}_{F_2}$ that is involutive, and as a map of lattices it is an involutive anti-automorphism of $F_2$.
\end{theorem}
\begin{proof}
On objects ${}^\perp$ is a bijection, since $(P^\perp)^\perp = P$.
It is contravariant on morphisms, because $P \leq Q$ implies $Q^\perp \leq P^\perp$, which is the morphism $Q^\perp \to P^\perp$, and it preserves identities and composition up to this reversal, since $P \leq Q \leq R$ implies $R^\perp \leq Q^\perp \leq P^\perp$.
As a map of lattices it sends meets to joins and joins to meets by the De Morgan laws, which hold in every orthomodular lattice, and it is bijective and order-reversing, hence an involutive anti-automorphism.
\end{proof}
We can now state what is special about the layer duality, in place of the vague claim that it is ``unique to $F_2$ among finite orthomodular lattices.''
The three pairings are the orthocomplementation of the $\MO_2$ factor, and this is rigid in the following sense.
\begin{proposition}[The layer duality is the orthocomplementation of $\MO_2$]
\label{prop:rigidity}
Let $L \cong \MO_2 \times B$ with $B$ a Boolean algebra, equipped with the product orthocomplementation, and stratify $L$ by the six contexts $0,a,b,a^\perp,b^\perp,1$ of the $\MO_2$ factor. Then orthocomplementation permutes the six context strata by the involution of $\MO_2$, that is, by the three transpositions $0 \leftrightarrow 1$, $a \leftrightarrow a^\perp$, and $b \leftrightarrow b^\perp$. The pair $0 \leftrightarrow 1$ exchanges the two central strata and the pairs $a \leftrightarrow a^\perp$ and $b \leftrightarrow b^\perp$ exchange the four non-central strata. For $B = 2^4$, that is for $L = F_2$, it is the vertical duality $B01 \leftrightarrow B06$ together with the horizontal dualities $B02 \leftrightarrow B05$ and $B03 \leftrightarrow B04$.
\end{proposition}
\begin{proof}
The product orthocomplementation acts as $(c,\beta) \mapsto (c^\perp,\beta^\perp)$, so it sends the stratum of context $c$ to the stratum of context $c^\perp$, and its action on the strata is the orthocomplementation of the factor $\MO_2$.
That orthocomplementation is part of the structure of $\MO_2$, since by definition $\MO_2$ is the modular lattice $M_4$ equipped with the orthocomplementation that interchanges the bounds, $0^\perp = 1$, and pairs the four atoms as $a \leftrightarrow a^\perp$ and $b \leftrightarrow b^\perp$.
It therefore realizes exactly the three transpositions $0 \leftrightarrow 1$, $a \leftrightarrow a^\perp$, $b \leftrightarrow b^\perp$.
By Proposition~\ref{prop:commute} the strata of contexts $0$ and $1$ are central and the four atomic strata are non-central, which gives the stated central and non-central pattern, and the case $B = 2^4$ is the statement about $F_2$.
\end{proof}
\begin{remark}[Rigidity is a property of the orthocomplemented structure]
\label{rem:rigidity}
The pairing is not determined by the bare lattice order.
The underlying modular lattice $M_4$ of $\MO_2$ admits three distinct orthocomplementations, one for each of the three ways of partitioning its four atoms into two complementary pairs, since any two distinct atoms of $M_4$ already have meet $0$ and join $1$.
What singles out the pairing $a \leftrightarrow a^\perp$, $b \leftrightarrow b^\perp$ is not the lattice $M_4$ but the orthocomplementation that $F_2$ carries as the free orthomodular lattice on $x$ and $y$, since in $F_2$ the generator $x$ has context $a$ and its orthocomplement $x^\perp$ has context $a^\perp$, and likewise $y$ and $y^\perp$ have contexts $b$ and $b^\perp$.
Since the isomorphism \eqref{eq:decomp-f2} is one of orthocomplemented lattices by Definition~\ref{def:operations}, it transports this orthocomplementation to the stated orthocomplementation of $\MO_2$, and the layer duality is rigid in that it is fixed by the orthocomplemented structure rather than chosen.
\end{remark}
\begin{corollary}[Order-two symmetry]
\label{cor:z2}
The group generated by ${}^\perp$, acting on $F_2$ as an involutive anti-automorphism, is $\bbZ_2$.
\end{corollary}
\begin{proof}
By Theorem~\ref{thm:functoriality} the map ${}^\perp$ is an anti-automorphism with ${}^\perp\circ{}^\perp = \mathrm{id}$, so it has order two and generates a group isomorphic to $\bbZ_2$.
\end{proof}
This $\bbZ_2$ is generated by an anti-automorphism and is not the full automorphism group of $F_2$, which is larger and contains, for example, the automorphism that exchanges the two generators $x$ and $y$.
Table~\ref{tab:duality} records the layer pairings, and Proposition~\ref{prop:rigidity} shows that this table is determined by the $\MO_2$ factor rather than being a contingent feature of $F_2$.
\begin{table}[htbp]
\caption{\label{tab:duality}The layer pairings induced by orthocomplementation. By Proposition~\ref{prop:rigidity} these pairings are exactly the orthocomplementation of the $\MO_2$ factor.}
\centering
\begin{tabular}{c c c c l}
\hline\hline
Layer & Context & Dual layer & Dual context & Pairing \\
\hline
$B01$ & $0$ (central) & $B06$ & $1$ (central) & vertical \\
$B02$ & $a$ & $B05$ & $a^\perp$ & horizontal \\
$B03$ & $b$ & $B04$ & $b^\perp$ & horizontal \\
$B06$ & $1$ (central) & $B01$ & $0$ (central) & vertical \\
\hline\hline
\end{tabular}
\end{table}
In a Boolean algebra orthocomplementation acts uniformly on all elements and produces a single duality pair.
In $F_2$ it acts on three pairs of layers and is context-dependent through the horizontal dualities, which encode the complementarity of the two generators.
The difference is that in $F_2$ the orthocomplement carries the $\MO_2$ stratification with it, so that negation rearranges the space of contexts and not only the truth-functional content.

\section{Classical Logic as a Context-Forgetting Projection}
\label{sec:projection}
The constructions of the previous sections converge on a single map, which is the main object of the paper.
It forgets the context factor of $F_2 \cong \MO_2 \times 2^4$ and keeps only the Boolean content,
\eq{
  \pi_{\mathrm{forget}}\colon F_2 \to 2^4, \qquad (c,\bvec b) \mapsto \bvec b. \label{eq:pi-forget}
}
The following proposition shows that this map turns the informal slogan ``classical logic is a quotient of quantum logic'' into a precise statement about orthocomplemented lattices.
\begin{proposition}[Classical logic as an orthocomplemented quotient]
\label{prop:boolean-quotient}
The map $\pi_{\mathrm{forget}}$ of Eq.~\eqref{eq:pi-forget} is a surjective homomorphism of orthocomplemented lattices, that is, it preserves $\wedge$, $\vee$, ${}^\perp$, the bottom, and the top. The relation $\sim$ defined by $P \sim P'$ if and only if $\pi_{\mathrm{forget}}(P) = \pi_{\mathrm{forget}}(P')$ is a lattice congruence, each of its classes has exactly six elements, and the quotient is the classical Boolean algebra,
\eq{
  F_2 / {\sim} \;\cong\; 2^4. \nonumber
}
\end{proposition}
\begin{proof}
By Definition~\ref{def:operations} the operations of $F_2$ act on the two factors of $\MO_2 \times 2^4$ independently, so the second projection satisfies $\pi_{\mathrm{forget}}(P_1 \wedge P_2) = \bvec b_1 \wedge \bvec b_2$, $\pi_{\mathrm{forget}}(P_1 \vee P_2) = \bvec b_1 \vee \bvec b_2$, and $\pi_{\mathrm{forget}}(P^\perp) = \neg\bvec b$.
It sends the bottom $0_{F_2} = (0,0000)$ to $0000$ and the top $1_{F_2} = (1,1111)$ to $1111$.
Hence $\pi_{\mathrm{forget}}$ is a homomorphism of orthocomplemented lattices, and it is surjective because every bit vector $\bvec b$ is attained, for instance by $(0,\bvec b) \in B01$.
The relation $\sim$ is the kernel of this homomorphism, and the kernel of a lattice homomorphism is a congruence, so $\sim$ is a congruence.
Its class through $(c,\bvec b)$ is the fiber $\pi_{\mathrm{forget}}^{-1}(\bvec b) = \{(c',\bvec b) : c' \in \MO_2\}$, which has $|\MO_2| = 6$ elements.
By the homomorphism theorem for lattices the quotient $F_2/{\sim}$ is isomorphic to the image of $\pi_{\mathrm{forget}}$, which is all of $2^4$.
\end{proof}
Proposition~\ref{prop:boolean-quotient} has three consequences that organize the rest of our reading.
First, classical logic is a six-to-one, information-losing image of the contextual calculus, because each of the sixteen Boolean operations is the image of exactly six elements of $F_2$, one for each context, and $96/16 = 6$.
Second, the projection is not invertible.
Distinct elements that share a bit vector, including elements in different layers, are identified, so no rule recovers the context from the Boolean record alone, and classical logic is strictly weaker than the contextual calculus.
Third, $F_2$ is a conservative extension of the Boolean fragment, since the kernel $B01$ is a copy of $2^4$ inside $F_2$, as a lattice, on which $\pi_{\mathrm{forget}}$ is a bijection, while $F_2$ also carries the inter-context operations of Theorem~\ref{thm:context-dynamics}, which have no counterpart in context-free classical logic.
The map $\pi_{\mathrm{forget}}$ is implicit in the decomposition $F_2 \cong \MO_2 \times 2^4$~\cite{beran1985} and in the decomposition theorem~\cite{ozawa2016,ozawa2021a,ozawa2026}, but it has not been isolated, named, and studied as an operation in its own right, which is the gap this section fills.

\subsection{The physical content of the forgetting projection}
\label{sec:physical}
Within quantum logic alone the projection \eqref{eq:pi-forget} is a purely formal operation.
It is a homomorphism of orthocomplemented lattices that deletes the context index, and nothing in the lattice structure by itself marks it as a physical process.
The projection acquires physical content only when one adjoins the probabilistic layer of quantum mechanics, namely the representation of an observable by a projection-valued measure and the Born statistical formula that assigns to a state and a projection the probability of the corresponding outcome.
If the propositions of the logic are realized as projections on a Hilbert space, after Birkhoff and von Neumann~\cite{birkhoff1936}, then forgetting the context corresponds to a projective measurement~\cite{ozawa1984} whose readout is the context-free Boolean record of which truth function obtains.
Within the measurement-strength framework~\cite{emori2026cc}, supported mathematically by a cognitive path-integral model~\cite{emori2026pi}, this context-forgetting projection corresponds to the strong-measurement limit.
It is the regime in which the measurement coupling to the neural environment carries a quantum-logical proposition onto a classical proposition, producing a definite and reportable outcome.
In cognitive and experimental terms, $\pi_{\mathrm{forget}}$ is the operation a task performs when it asks for one reportable answer.
The response record keeps the Boolean outcome, such as yes or no, left or right, or face or non-face, but it does not keep the context that made the outcome meaningful, namely the question order, the task instruction, the attentional set, or the measurement basis selected by the experiment.
Classical logic is therefore adequate for the final record but not for the generative process that produces it.
The forgetting projection is an algebraic counterpart of Bohr's reading of complementarity~\cite{bohr1928,bohr1937, bohr1958}, on which observables have no fixed meaning independent of an experimental context, so that recovering a context-free Boolean record is the act of fixing a context by measurement.
The contextual hidden-variable theories characterized through quantum set theory~\cite{ozawa2022hv} sit in the same conceptual place, because there too a classical, context-free description is available on the commutative part.

This reading is consistent with the transfer principle of quantum set theory~\cite{ozawa2021a}, which bounds the quantum truth value of a theorem of set theory from below by the commutator,
\eq{
  [\![\,\phi(u_1,\dots,u_n)\,]\!]_Q \;\geq\; \com(u_1,\dots,u_n), \label{eq:transfer}
}
where $[\![\,\phi(u_1,\dots,u_n)\,]\!]_Q$ is the truth value of $\phi(u_1,\dots,u_n)$ in the quantum-set-theoretic universe $V^{(Q)}$.
Classical reasoning is valid to the extent that the commutator is large, that is, on the commutative part, and the forgetting projection selects this classical face at the level of the two-generator logic.
Contextual superposition, embodied by the layer $B06$, is already present in $F_2$, while entanglement is not.
Entanglement requires contextuality together with the additional probabilistic and tensor-product structure of a Hilbert-space description, which lies beyond the present formal analysis \cite{gunji2026b}.
In one sentence, classical logic is the shadow that a contextual calculus casts when the context is measured away.

\section{Discussion}
\label{sec:discussion}
We relate the construction to three neighboring frameworks, namely Takeuti's three logics, the measurement-strength account of cognition, and the sheaf- and topos-theoretic treatments of contextuality.
Each comparison is structural, and we mark the points where a tempting identification fails.

\subsection{Takeuti's three logics, and a quotient that does not exist}
\label{sec:takeuti}
This subsection blocks a tempting but incorrect simplification of the measurement-strength reading, namely that the intermediate regime could be obtained by quotienting the contextual logic onto a three-valued distributive chain.
Intermediate measurement is not a halfway Boolean logic, since it requires selecting a context without forgetting it entirely, and we show that the natural chain-valued image does not exist as a lattice quotient.

Takeuti~\cite{takeuti1981b} distinguishes three kinds of logic, which he calls God's logic, Human's logic, and the logic of things, and associates the third with orthomodular structure.
The first two correspond to objects in our setting.
God's logic, the classical case, is the Boolean image $2^4$ obtained after $\pi_{\mathrm{forget}}$, in which a unique Boolean complement exists.
The logic of things is the full lattice $F_2$ with its inter-context operations and its non-distributive behavior across contexts.
Human's logic, the intuitionistic case, is meant to be an intermediate, distributive regime in which inference proceeds inside one fixed context at a time.

It is tempting to model this intermediate regime by merging the four single contexts $a,b,a^\perp,b^\perp$ into one value $m$, so as to obtain a three-element chain $C_3 = \{0 < m < 1\}$ and an image of the form $C_3 \times 2^4$ carrying a Heyting structure.
This cannot be done as a lattice quotient of $F_2$.
The factor $\MO_2$, which is the modular ortholattice $M_4$, is a simple lattice, so its only congruences are the identity and the total congruence, and its only quotients are $\MO_2$ itself and the one-element lattice.
Lattices form a congruence-distributive variety, in which a congruence on a finite direct product is a product of congruences on the factors, so every quotient of $F_2 \cong \MO_2 \times 2^4$ has the form (quotient of $\MO_2$) $\times$ (quotient of $2^4$).
The first factor is either $\MO_2$ or trivial, and is never the three-element chain $C_3$.
Indeed, since $\MO_2$ is simple, its only lattice quotient other than the one-element lattice is $\MO_2$ itself, which is non-distributive, so $\MO_2$ admits no non-trivial homomorphism onto a distributive lattice at all, and even its largest distributive quotient collapses it entirely.
We therefore do not present $C_3 \times 2^4$ as a quotient of $F_2$.

A Heyting or bi-Heyting image, if one is wanted, should be built not by quotienting the simple factor $\MO_2$ but by a different route, for instance by passing to the poset of contexts and the distributive lattice of its down-sets, on which a relative pseudocomplement is available, or by a homomorphism from $F_2$ onto a suitable distributive lattice.
We leave a rigorous construction to future work and treat the intermediate regime here as a heuristic analogy whose precise form is not yet fixed.
Within this guarded reading our analysis also locates Takeuti's observation that intuitionistic logic, although it looks more distant from classical logic than quantum logic does, is in a structural sense closer.
Intuitionistic and classical logic are both distributive as lattices, while $F_2$ is non-distributive because its operations update a context index in addition to the Boolean content, as in $(a,1100) \wedge (b,1010) = (0,1000)$, where a context collapse accompanies the Boolean meet.
A distributive image has an internal implication and no separate context being updated, which is the respect in which it stays closer to the Boolean case than $F_2$ does.

\subsection{A reading through measurement strength}
\label{sec:meas-strength}
The layer picture gives a vocabulary, independent of any physical realization, for how the operations preserve or discard context, and it makes contact with the measurement-strength program~\cite{emori2026cc,emori2026pi,ozawa2023,fuyama2025}.
As a guide to translation, and not as an identification, one may distinguish three regimes.
A weak, context-preserving regime keeps several incompatible candidates in play, and is modeled by remaining in, or returning to, the central layer $B06$.
An intermediate regime selects a single-context viewpoint while keeping non-Boolean structure, and is modeled by a transition from $B06$ to one of $B02$--$B05$, in which the context is narrowed but not eliminated.
A strong regime discards the remaining context, and the content is then expressed by its image under the context-forgetting projection $\pi_{\mathrm{forget}}$ of Section~\ref{sec:projection}, that is, by a class of the quotient $F_2/{\sim}\,\cong\,2^4$.
The kernel $B01$ is a canonical transversal of this quotient, one representative per class, and $\pi_{\mathrm{forget}}$ restricts to a lattice isomorphism $B01 \to 2^4$.
This restriction is not an isomorphism of orthocomplemented lattices, since $B01$ is not closed under ${}^\perp$, the orthocomplement of $(0,\bvec b)$ being $(1,\neg\bvec b) \in B06$ by Theorem~\ref{thm:vertical}.
The strong regime is therefore the Boolean quotient itself, of which $B01$ furnishes the canonical representatives, rather than an orthocomplemented sublayer.
The comparison shows that strength need not be a purely quantitative notion, because the selection of a context and the elimination of a context are different kinds of information loss.
At the level of quantum set theory the transfer principle \eqref{eq:transfer} makes the same point from the semantic side, since classical theorems hold where the commutator is maximal and only weaker, context-dependent truth values are available elsewhere.

\subsection{Relation to sheaf- and topos-theoretic approaches}
\label{sec:sheaf-topos}
We have studied the internal structure of $F_2$ and its decomposition, without treating the global gluing of such local pieces over a full measurement cover.
This local role is complementary to the sheaf- and topos-theoretic approaches to contextuality, in which the commutative subalgebras of a von Neumann algebra form a measurement cover and contextuality is the absence of a global section of the associated spectral presheaf~\cite{isham1998,doring2008,abramsky2011}.
By the universal property of the free orthomodular lattice, the two-generator fragment $\langle p,q\rangle$ of any orthomodular lattice is a quotient of $F_2$, and in the simplest case the $\MO_2$ factor is the horizontal sum, equivalently the pushout along the two-element chain, of the two single-generator Boolean contexts.
This construction is classical in the orthomodular literature~\cite{greechie1971,kalmbach1983} and has recently been recast as a left adjoint by Gunji \textit{et al.}~\cite{gunji2026a}.
In this way $F_2$ is a finite and explicitly computable local model, whose global organization is treated on the set-theoretic side by the commutator and the transfer principle of quantum set theory, and on the sheaf-theoretic side by measurement covers and spectral presheaves.
The relation between these two global descriptions is itself the subject of recent work bridging quantum set theory and topos quantum theory~\cite{doring2021}.

\section{Conclusion}
\label{sec:conclusion}
These results present the logical component of a three-paper measurement-strength program, in which the present paper supplies the operation by which a contextual cognitive process becomes a classical, reportable record.
We have made the product presentation $F_2 \cong \MO_2 \times 2^4$ operational through an explicit context--bit-vector calculus, and we have isolated the context-forgetting projection $\pi_{\mathrm{forget}}$ as the operation that turns the contextual logic into classical logic.
The six layers of the lattice are classified by commutativity in Proposition~\ref{prop:commute}, which identifies the Boolean kernel $B01$ and the central superposition layer $B06$ as the two central layers and the four single-context layers $B02$--$B05$ as the non-central ones.
The dualities among the layers are exactly the orthocomplementation of the $\MO_2$ factor, by Proposition~\ref{prop:rigidity}, and are therefore rigid rather than contingent.
The main result, Proposition~\ref{prop:boolean-quotient}, is that $\pi_{\mathrm{forget}}$ is a surjective homomorphism of orthocomplemented lattices whose quotient is the classical Boolean algebra $2^4$, so that classical logic is a uniform six-to-one image of the contextual calculus, and we have explained how this forgetting acquires physical content once the Born statistical formula and projective measurement are adjoined.
We have also shown that Takeuti's intermediate, intuitionistic regime cannot be realized as a lattice quotient, because the relevant factor $\MO_2$ is simple, and we have therefore kept that regime as a heuristic analogy.

The broader reading these results support is that, in the two-generator free setting, classical truth-functional structure is an information-losing image of a richer context-sensitive calculus, so that contextuality is the generic situation and the classical regime is a limit.
Whenever several irreducible contexts coexist one should expect orthomodular features rather than treat them as anomalies, a stance consistent with the characterization of contextual hidden-variable theories in quantum set theory~\cite{ozawa2022hv}.

The results also fix what the present paper contributes to the companion works.
Ref.~\cite{emori2026pi} supplies the dynamics along the measurement-strength axis, and Ref.~\cite{emori2026cc} applies that axis to theories of consciousness and to experimental protocols.
The context-forgetting projection isolated here is the logical form of the strong-measurement limit used in both, so a reader who begins with the present paper has, in $\pi_{\mathrm{forget}}$, a reason to read the dynamical and empirical companions.
Future work includes a systematic categorical formulation, a rigorous construction of the intermediate distributive regime by a route other than quotienting the simple factor $\MO_2$, and a tighter interface between $F_2$ as a formal context skeleton and the commutator control of quantum set theory, in which states, observables, and operations enter explicitly \cite{ozawa2023}.

\begin{acknowledgments}
We are grateful to Masanao Ozawa for valuable discussions.
We thank the organizers and participants at the Ernst Str{\"{u}}ngmann Forum ``Simplicity behind Absurdity: The Power of Quantum Thinking'' held in Frankfurt in September 2025, where this work was initiated.
This work was supported by JST ASPIRE Grant Number JPMJAP2318, JSPS KAKENHI Grant Numbers 24H02200 and 26H02539, JST SPRING Grant Number JPMJSP2119, and the RIKEN Junior Research Associate Program.
\end{acknowledgments}

\appendix

\section{Complete Enumeration of \texorpdfstring{$F_2$}{F2}}
\label{app:enumeration}
This appendix lists all ninety-six elements of $F_2$, organized into the six layers $B01$--$B06$.
Each element is given by its Beran number~\cite{beran1985}, by an ortholattice polynomial in $x,y,x^\perp,y^\perp$, and by its vector $(c,b_1b_2b_3b_4)$ with context $c \in \{0,a,b,a^\perp,b^\perp,1\}$ and $(b_1,b_2,b_3,b_4) \in 2^4$.
The note column flags the named two-variable operation an element realizes: the Sasaki ($\mathrm S$), contrapositive Sasaki ($\mathrm{CPS}$), and relevance ($\mathrm R$) conditionals, the corresponding dual conjunctions $\dcS$, $\dcC$, $\dcR$, the two further Kotas--Ozawa operations $\Rightarrow_1$, $\Rightarrow_4$, the generators, the layer tops $\top_{B0k}$, and the global bounds $0_{F_2}$, $1_{F_2}$.
An element with no such name, but whose bit vector matches that of conjunction or disjunction, is marked $\wedge$ or $\vee$.
The four bits record the value of the element on the four atoms of the Boolean subalgebra generated by $x$ and $y$, in the order
\eq{
  (b_1,b_2,b_3,b_4) \;\longleftrightarrow\; (x \wedge y,\ x \wedge y^\perp,\ x^\perp \wedge y,\ x^\perp \wedge y^\perp), \nonumber
}
which is the order consistent with every vector below.
Each triple of Beran number, formula, and vector has been verified by evaluating the formula in the model $F_2 \cong \MO_2 \times 2^4$ of Definition~\ref{def:operations}, and the ninety-six vectors are pairwise distinct and exhaust the product $\MO_2 \times 2^4$.

\subsection{\texorpdfstring{Layer $B01$: Boolean kernel (context $0$)}{Layer B01: Boolean kernel (context 0)}}
\label{app:b01}
\begin{longtable}{c l c l}
\hline\hline
Beran No. & Formula & Vector & Note \\
\hline
\endfirsthead
\hline\hline
Beran No. & Formula & Vector & Note \\
\hline
\endhead
\hline\hline
\endfoot
01 & $0 = x \wedge y^\perp \wedge x^\perp \wedge y$ & $(0,0000)$ & $0_{F_2}$ \\
02 & $x \wedge y$ & $(0,1000)$ & $\wedge$ \\
03 & $x \wedge y^\perp$ & $(0,0100)$ & \\
04 & $x^\perp \wedge y$ & $(0,0010)$ & \\
05 & $x^\perp \wedge y^\perp$ & $(0,0001)$ & \\
06 & $(x \wedge y) \vee (x \wedge y^\perp)$ & $(0,1100)$ & \\
07 & $(x \wedge y) \vee (x^\perp \wedge y)$ & $(0,1010)$ & \\
08 & $(x \wedge y) \vee (x^\perp \wedge y^\perp)$ & $(0,1001)$ & \\
09 & $(x \wedge y^\perp) \vee (x^\perp \wedge y)$ & $(0,0110)$ & \\
10 & $(x^\perp \wedge y^\perp) \vee (x \wedge y^\perp)$ & $(0,0101)$ & \\
11 & $(x^\perp \wedge y^\perp) \vee (x^\perp \wedge y)$ & $(0,0011)$ & \\
12 & $(x \wedge y) \vee (x \wedge y^\perp) \vee (x^\perp \wedge y)$ & $(0,1110)$ & $\vee$ \\
13 & $(x \wedge y) \vee (x \wedge y^\perp) \vee (x^\perp \wedge y^\perp)$ & $(0,1101)$ & \\
14 & $(x \wedge y) \vee (x^\perp \wedge y) \vee (x^\perp \wedge y^\perp)$ & $(0,1011)$ & $\mathrm R$ \\
15 & $(x \wedge y^\perp) \vee (x^\perp \wedge y) \vee (x^\perp \wedge y^\perp)$ & $(0,0111)$ & \\
16 & $(x \wedge y) \vee (x \wedge y^\perp) \vee (x^\perp \wedge y) \vee (x^\perp \wedge y^\perp)$ & $(0,1111)$ & $\top_{B01}$ \\
\end{longtable}

\subsection{\texorpdfstring{Layer $B02$: context $a = x \wedge \com(x,y)^\perp$}{Layer B02: context a}}
\label{app:b02}
\begin{longtable}{c l c l}
\hline\hline
Beran No. & Formula & Vector & Note \\
\hline
\endfirsthead
\hline\hline
Beran No. & Formula & Vector & Note \\
\hline
\endhead
\hline\hline
\endfoot
17 & $x \wedge (x^\perp \vee y) \wedge (x^\perp \vee y^\perp)$ & $(a,0000)$ & \\
18 & $x \wedge (x^\perp \vee y)$ & $(a,1000)$ & $\dcS$ \\
19 & $x \wedge (x^\perp \vee y^\perp)$ & $(a,0100)$ & \\
20 & $(x^\perp \wedge y) \vee \bigl(x \wedge (x^\perp \vee y) \wedge (x^\perp \vee y^\perp)\bigr)$ & $(a,0010)$ & \\
21 & $(x^\perp \wedge y^\perp) \vee \bigl(x \wedge (x^\perp \vee y) \wedge (x^\perp \vee y^\perp)\bigr)$ & $(a,0001)$ & \\
22 & $x$ & $(a,1100)$ & generator \\
23 & $(x^\perp \vee y) \wedge \bigl(x \vee (x^\perp \wedge y)\bigr)$ & $(a,1010)$ & \\
24 & $(x^\perp \vee y) \wedge \bigl(x \vee (x^\perp \wedge y^\perp)\bigr)$ & $(a,1001)$ & \\
25 & $(x^\perp \vee y^\perp) \wedge \bigl(x \vee (x^\perp \wedge y)\bigr)$ & $(a,0110)$ & \\
26 & $(x^\perp \vee y^\perp) \wedge \bigl(x \vee (x^\perp \wedge y^\perp)\bigr)$ & $(a,0101)$ & \\
27 & $(x^\perp \vee y^\perp) \wedge (x^\perp \vee y) \wedge \bigl(x \vee (x^\perp \wedge y^\perp) \vee (x^\perp \wedge y)\bigr)$ & $(a,0011)$ & \\
28 & $x \vee (x^\perp \wedge y)$ & $(a,1110)$ & $\vee$ \\
29 & $x \vee (x^\perp \wedge y^\perp)$ & $(a,1101)$ & \\
30 & $(x^\perp \vee y) \wedge \bigl(x \vee (x^\perp \wedge y) \vee (x^\perp \wedge y^\perp)\bigr)$ & $(a,1011)$ & $\Rightarrow_1$ \\
31 & $(x^\perp \vee y^\perp) \wedge \bigl(x \vee (x^\perp \wedge y) \vee (x^\perp \wedge y^\perp)\bigr)$ & $(a,0111)$ & \\
32 & $x \vee (x^\perp \wedge y) \vee (x^\perp \wedge y^\perp)$ & $(a,1111)$ & $\top_{B02}$ \\
\end{longtable}

\subsection{\texorpdfstring{Layer $B03$: context $b = y \wedge \com(x,y)^\perp$}{Layer B03: context b}}
\label{app:b03}
\begin{longtable}{c l c l}
\hline\hline
Beran No. & Formula & Vector & Note \\
\hline
\endfirsthead
\hline\hline
Beran No. & Formula & Vector & Note \\
\hline
\endhead
\hline\hline
\endfoot
33 & $y \wedge (y^\perp \vee x) \wedge (y^\perp \vee x^\perp)$ & $(b,0000)$ & \\
34 & $y \wedge (y^\perp \vee x)$ & $(b,1000)$ & $\dcC$ \\
35 & $(x \wedge y^\perp) \vee \bigl(y \wedge (y^\perp \vee x^\perp) \wedge (y^\perp \vee x)\bigr)$ & $(b,0100)$ & \\
36 & $y \wedge (y^\perp \vee x^\perp)$ & $(b,0010)$ & \\
37 & $(x^\perp \wedge y^\perp) \vee \bigl(y \wedge (y^\perp \vee x) \wedge (y^\perp \vee x^\perp)\bigr)$ & $(b,0001)$ & \\
38 & $(x \vee y^\perp) \wedge \bigl(y \vee (y^\perp \wedge x)\bigr)$ & $(b,1100)$ & \\
39 & $y$ & $(b,1010)$ & generator \\
40 & $(x \vee y^\perp) \wedge \bigl(y \vee (y^\perp \wedge x^\perp)\bigr)$ & $(b,1001)$ & \\
41 & $(x^\perp \vee y^\perp) \wedge \bigl(y \vee (x \wedge y^\perp)\bigr)$ & $(b,0110)$ & \\
42 & $(x^\perp \vee y^\perp) \wedge (x \vee y^\perp) \wedge \bigl(y \vee (x^\perp \wedge y^\perp) \vee (x \wedge y^\perp)\bigr)$ & $(b,0101)$ & \\
43 & $(x^\perp \vee y^\perp) \wedge \bigl(y \vee (y^\perp \wedge x^\perp)\bigr)$ & $(b,0011)$ & \\
44 & $y \vee (x \wedge y^\perp)$ & $(b,1110)$ & $\vee$ \\
45 & $(x \vee y^\perp) \wedge \bigl(y \vee (y^\perp \wedge x) \vee (y^\perp \wedge x^\perp)\bigr)$ & $(b,1101)$ & \\
46 & $y \vee (y^\perp \wedge x^\perp)$ & $(b,1011)$ & $\mathrm{CPS}$ \\
47 & $(x^\perp \vee y^\perp) \wedge \bigl(y \vee (y^\perp \wedge x) \vee (y^\perp \wedge x^\perp)\bigr)$ & $(b,0111)$ & \\
48 & $y \vee (y^\perp \wedge x) \vee (y^\perp \wedge x^\perp)$ & $(b,1111)$ & $\top_{B03}$ \\
\end{longtable}

\subsection{\texorpdfstring{Layer $B04$: context $b^\perp = y^\perp \wedge \com(x,y)^\perp$}{Layer B04: context b-perp}}
\label{app:b04}
\begin{longtable}{c l c l}
\hline\hline
Beran No. & Formula & Vector & Note \\
\hline
\endfirsthead
\hline\hline
Beran No. & Formula & Vector & Note \\
\hline
\endhead
\hline\hline
\endfoot
49 & $y^\perp \wedge (x^\perp \vee y) \wedge (x \vee y)$ & $(b^\perp,0000)$ & \\
50 & $(x \wedge y) \vee \bigl(y^\perp \wedge (x^\perp \vee y) \wedge (x \vee y)\bigr)$ & $(b^\perp,1000)$ & $\wedge$ \\
51 & $y^\perp \wedge (x \vee y)$ & $(b^\perp,0100)$ & \\
52 & $(x^\perp \wedge y) \vee \bigl(y^\perp \wedge (x^\perp \vee y) \wedge (x \vee y)\bigr)$ & $(b^\perp,0010)$ & \\
53 & $y^\perp \wedge (x^\perp \vee y)$ & $(b^\perp,0001)$ & \\
54 & $(x \vee y) \wedge \bigl(y^\perp \vee (x \wedge y)\bigr)$ & $(b^\perp,1100)$ & \\
55 & $(x \vee y) \wedge (x^\perp \vee y) \wedge \bigl(y^\perp \vee (x \wedge y) \vee (x^\perp \wedge y)\bigr)$ & $(b^\perp,1010)$ & \\
56 & $(x^\perp \vee y) \wedge \bigl(y^\perp \vee (x \wedge y)\bigr)$ & $(b^\perp,1001)$ & \\
57 & $(x \vee y) \wedge \bigl(y^\perp \vee (x^\perp \wedge y)\bigr)$ & $(b^\perp,0110)$ & \\
58 & $y^\perp$ & $(b^\perp,0101)$ & generator \\
59 & $(x^\perp \vee y) \wedge \bigl(y^\perp \vee (x^\perp \wedge y)\bigr)$ & $(b^\perp,0011)$ & \\
60 & $(x \vee y) \wedge \bigl(y^\perp \vee (y \wedge x^\perp) \vee (y \wedge x)\bigr)$ & $(b^\perp,1110)$ & $\vee$ \\
61 & $y^\perp \vee (x \wedge y)$ & $(b^\perp,1101)$ & \\
62 & $(x^\perp \vee y) \wedge \bigl(y^\perp \vee (x^\perp \wedge y) \vee (x \wedge y)\bigr)$ & $(b^\perp,1011)$ & $\Rightarrow_4$ \\
63 & $y^\perp \vee (y \wedge x^\perp)$ & $(b^\perp,0111)$ & \\
64 & $y^\perp \vee (x^\perp \wedge y) \vee (x \wedge y)$ & $(b^\perp,1111)$ & $\top_{B04}$ \\
\end{longtable}

\subsection{\texorpdfstring{Layer $B05$: context $a^\perp = x^\perp \wedge \com(x,y)^\perp$}{Layer B05: context a-perp}}
\label{app:b05}
\begin{longtable}{c l c l}
\hline\hline
Beran No. & Formula & Vector & Note \\
\hline
\endfirsthead
\hline\hline
Beran No. & Formula & Vector & Note \\
\hline
\endhead
\hline\hline
\endfoot
65 & $x^\perp \wedge (x \vee y^\perp) \wedge (x \vee y)$ & $(a^\perp,0000)$ & \\
66 & $(x \wedge y) \vee \bigl(x^\perp \wedge (x \vee y^\perp) \wedge (x \vee y)\bigr)$ & $(a^\perp,1000)$ & $\wedge$ \\
67 & $(x \wedge y^\perp) \vee \bigl(x^\perp \wedge (x \vee y^\perp) \wedge (x \vee y)\bigr)$ & $(a^\perp,0100)$ & \\
68 & $x^\perp \wedge (x \vee y)$ & $(a^\perp,0010)$ & \\
69 & $x^\perp \wedge (x \vee y^\perp)$ & $(a^\perp,0001)$ & \\
70 & $(x \vee y) \wedge (x \vee y^\perp) \wedge \bigl(x^\perp \vee (x \wedge y) \vee (x \wedge y^\perp)\bigr)$ & $(a^\perp,1100)$ & \\
71 & $(x \vee y) \wedge \bigl(x^\perp \vee (x \wedge y)\bigr)$ & $(a^\perp,1010)$ & \\
72 & $(x \vee y^\perp) \wedge \bigl(x^\perp \vee (x \wedge y)\bigr)$ & $(a^\perp,1001)$ & \\
73 & $(x \vee y) \wedge \bigl(x^\perp \vee (x \wedge y^\perp)\bigr)$ & $(a^\perp,0110)$ & \\
74 & $(x \vee y^\perp) \wedge \bigl(x^\perp \vee (x \wedge y^\perp)\bigr)$ & $(a^\perp,0101)$ & \\
75 & $x^\perp$ & $(a^\perp,0011)$ & generator \\
76 & $(x \vee y) \wedge \bigl(x^\perp \vee (x \wedge y^\perp) \vee (x \wedge y)\bigr)$ & $(a^\perp,1110)$ & $\vee$ \\
77 & $(x \vee y^\perp) \wedge \bigl(x^\perp \vee (x \wedge y^\perp) \vee (x \wedge y)\bigr)$ & $(a^\perp,1101)$ & \\
78 & $x^\perp \vee (x \wedge y)$ & $(a^\perp,1011)$ & $\mathrm S$ \\
79 & $x^\perp \vee (x \wedge y^\perp)$ & $(a^\perp,0111)$ & \\
80 & $x^\perp \vee (x \wedge y^\perp) \vee (x \wedge y)$ & $(a^\perp,1111)$ & $\top_{B05}$ \\
\end{longtable}

\subsection{\texorpdfstring{Layer $B06$: context $1 = \com(x,y)^\perp$}{Layer B06: context 1}}
\label{app:b06}
\begin{longtable}{c l c l}
\hline\hline
Beran No. & Formula & Vector & Note \\
\hline
\endfirsthead
\hline\hline
Beran No. & Formula & Vector & Note \\
\hline
\endhead
\hline\hline
\endfoot
81 & $(x \vee y) \wedge (x \vee y^\perp) \wedge (x^\perp \vee y^\perp) \wedge (x^\perp \vee y)$ & $(1,0000)$ & $\com(x,y)^\perp$ \\
82 & $(x^\perp \vee y) \wedge (x \vee y^\perp) \wedge (x \vee y)$ & $(1,1000)$ & $\dcR$ \\
83 & $(x^\perp \vee y^\perp) \wedge (x \vee y^\perp) \wedge (x \vee y)$ & $(1,0100)$ & \\
84 & $(x^\perp \vee y^\perp) \wedge (x^\perp \vee y) \wedge (x \vee y)$ & $(1,0010)$ & \\
85 & $(x^\perp \vee y^\perp) \wedge (x^\perp \vee y) \wedge (x \vee y^\perp)$ & $(1,0001)$ & \\
86 & $(x \vee y) \wedge (x \vee y^\perp)$ & $(1,1100)$ & \\
87 & $(x \vee y) \wedge (x^\perp \vee y)$ & $(1,1010)$ & \\
88 & $(x^\perp \vee y) \wedge (x \vee y^\perp)$ & $(1,1001)$ & \\
89 & $(x^\perp \vee y^\perp) \wedge (x \vee y)$ & $(1,0110)$ & \\
90 & $(x^\perp \vee y^\perp) \wedge (x \vee y^\perp)$ & $(1,0101)$ & \\
91 & $(x^\perp \vee y^\perp) \wedge (x^\perp \vee y)$ & $(1,0011)$ & \\
92 & $x \vee y$ & $(1,1110)$ & $\vee$ \\
93 & $x \vee y^\perp$ & $(1,1101)$ & \\
94 & $x^\perp \vee y$ & $(1,1011)$ & $x \to y$ \\
95 & $x^\perp \vee y^\perp$ & $(1,0111)$ & \\
96 & $1$ & $(1,1111)$ & $1_{F_2}$ \\
\end{longtable}
In layer $B06$ the context $\com(x,y)^\perp$ is the top of the $\MO_2$ factor, and the global top $1_{F_2} = (1,1111)$ resides here, together with $\com(x,y)^\perp = (1,0000)$, which is the bottom of this layer but not the global bottom, since the global bottom $0_{F_2} = (0,0000)$ lies in $B01$ (No.~01).

\section{Worked Computations and Verification}
\label{app:examples}
This appendix collects worked examples of the context transitions and records the computational verification used in the main text.
Every computation uses the component-by-component operations of Definition~\ref{def:operations}.

\paragraph*{Meet and context collapse.}
A cross-layer meet collapses to the kernel.
Take $P_1 = (a,1100)$, which is No.~22, the generator $x$ in $B02$, and $P_2 = (b,1010)$, which is No.~39, the generator $y$ in $B03$.
The context meet is $a \wedge_{\MO_2} b = 0$, because $a$ and $b$ are distinct atoms, and the bit meet is $1100 \wedge 1010 = 1000$, so $P_1 \wedge P_2 = (0,1000)$, which is No.~02, the conjunction $x \wedge y$ in $B01$.
A within-layer meet preserves the layer.
For $P_1 = (a,1100)$ and $P_2 = (a,0101)$ one has $a \wedge_{\MO_2} a = a$ and $1100 \wedge 0101 = 0100$, so $P_1 \wedge P_2 = (a,0100)$ remains in $B02$.

\paragraph*{Join and context expansion.}
A cross-layer join ascends to $B06$.
For $P_1 = (a,1000)$ and $P_2 = (b,0100)$ the context join is $a \vee_{\MO_2} b = \com(x,y)^\perp = 1$, because distinct atoms join to the top of $\MO_2$, and the bit join is $1000 \vee 0100 = 1100$, so $P_1 \vee P_2 = (1,1100)$ in $B06$.
A join of the neutral context with a single context adopts the single context.
For $P_1 = (0,0110)$ and $P_2 = (a,1000)$ one has $0 \vee_{\MO_2} a = a$ and $0110 \vee 1000 = 1110$, so $P_1 \vee P_2 = (a,1110)$ in $B02$.

\paragraph*{Orthocomplement and the dualities.}
For the vertical duality take $P = (0,1010)$ in $B01$.
The context complement is $0^\perp = \com(x,y)^\perp = 1$ and the bit complement is $\neg 1010 = 0101$, so $P^\perp = (1,0101)$ in $B06$.
For the horizontal duality take $P = (a,1100)$, which is No.~22, the generator $x$.
The context complement is $a^\perp$ and the bit complement is $\neg 1100 = 0011$, so $P^\perp = (a^\perp,0011)$, which is No.~75, the complement $x^\perp$ in $B05$.
The same computation with $b$ in place of $a$ gives the pairing of $B03$ with $B04$ through $y \leftrightarrow y^\perp$.

\paragraph*{The commutator of a cross-layer pair.}
For $P = (a,1100)$ and $Q = (b,1010)$ we compute the central commutator from the two-element closed form of Eq.~\eqref{eq:com-two}, which by Theorem~2.1(i) of Ozawa~\cite{ozawa2026} is the central commutator itself and not a separate construction.
We have $P \wedge Q = (0,1000)$, $Q^\perp = (b^\perp,0101)$ so $P \wedge Q^\perp = (0,0100)$, $P^\perp = (a^\perp,0011)$ so $P^\perp \wedge Q = (0,0010)$, and $P^\perp \wedge Q^\perp = (0,0001)$.
Their join is $(0,1111)$, which is $1_{B01}$, the top of the kernel.
Therefore $\com(x,y) = (0,1111)$ and $\com(x,y)^\perp = (1,0000)$.
The two generators reach maximal commutativity inside the Boolean kernel but do not commute globally, in agreement with Proposition~\ref{prop:commute}.

\paragraph*{Verification.}
We enumerated all ninety-six elements of $F_2$ as the Cartesian product of the six contexts of $\MO_2$ with the sixteen bit vectors, and evaluated every formula of Appendix~\ref{app:enumeration} in the model by the component-by-component rules of Definition~\ref{def:operations}.
Three checks were carried out.
First, each formula reproduces the vector listed beside it, and the ninety-six vectors are pairwise distinct and exhaust the product $\MO_2 \times 2^4$.
Second, the orthomodular law, that $P \leq Q$ implies $Q = P \vee (Q \wedge P^\perp)$, holds for every ordered pair among the ninety-six elements.
Third, the six Kotas--Ozawa conditionals~\cite{kotas1967,ozawa2026} are reproduced, with $\Rightarrow_1$ and $\Rightarrow_4$ landing on No.~30 and No.~62 as stated in Section~\ref{sec:conditionals}, and the Sasaki, contrapositive Sasaki, and relevance conditionals landing on No.~78, No.~46, and No.~14.
Table~\ref{tab:transition-rules} summarizes the resulting context-transition rules.
\begin{table}[htbp]
\caption{\label{tab:transition-rules}Context-transition rules, read off from the lattice structure of $\MO_2$. Here $Bxx$ denotes any single-context layer and $-$ denotes any bit vector.}
\centering
\begin{tabular}{l c c l}
\hline\hline
Operation & Contexts $(c_1,c_2)$ & Result context & Layer transition \\
\hline
$P_1 \wedge P_2$ & $(0,c)$ & $0$ & $B01 \wedge \text{anything} \to B01$ \\
$P_1 \wedge P_2$ & $(a,b)$ & $0$ & $B02 \wedge B03 \to B01$ \\
$P_1 \wedge P_2$ & $(a,a)$ & $a$ & $B02 \wedge B02 \to B02$ \\
$P_1 \vee P_2$ & $(0,c)$ & $c$ & $B01 \vee Bxx \to Bxx$ \\
$P_1 \vee P_2$ & $(a,b)$ & $1$ & $B02 \vee B03 \to B06$ \\
$P_1 \vee P_2$ & $(a,a^\perp)$ & $1$ & $B02 \vee B05 \to B06$ \\
$P^\perp$ & $(0,-)$ & $1$ & $B01 \to B06$ \\
$P^\perp$ & $(a,-)$ & $a^\perp$ & $B02 \to B05$ \\
$P^\perp$ & $(b,-)$ & $b^\perp$ & $B03 \to B04$ \\
\hline\hline
\end{tabular}
\end{table}

\bibliography{ref}

\end{document}